\documentclass[preprint,authoryear,11pt]{elsarticle}

\usepackage{amsmath,amsthm,amssymb}
\usepackage{rotating}
\usepackage{morefloats}
\usepackage{comment}
\usepackage{natbib}
\usepackage{caption}
\usepackage{booktabs}
\usepackage{multirow}
\usepackage{longtable}
\usepackage{setspace}
\usepackage[margin=2.5cm]{geometry}
\usepackage{bbm}
\usepackage{pdflscape}

\newtheorem{prp}{Proposition}

\newtheorem{crl}{Corollary}

\newcommand{\w}{\boldsymbol{w}}
\newcommand{\wt}{\boldsymbol{w}^{\!\top}}
\newcommand{\tw}{\boldsymbol{\tilde{w}}}

\begin{document}
\title{Entropic Value-at-Risk parity for tempered stable returns}
\author{Jaehyung Choi\corref{cor}}
\ead{jj.jaehyung.choi@gmail.com}
\cortext[cor]{Corresponding author}

\begin{abstract}
	We develop Entropic Value-at-Risk (EVaR) parity for tempered stable returns. EVaR-based inverse risk parity (IRP) and equal risk contribution (ERC) portfolios are constructed using multivariate normal tempered stable models and independent component analysis with tempered stable components. We derive the corresponding asset-level EVaR and EVaR-deviation contributions and use the latter to separate the fitted location term from EVaR risk contributions. Under Gaussian returns, EVaR-deviation IRP and ERC recover conventional volatility IRP and ERC weights. We evaluate the resulting portfolios in three investment universes. Empirically, EVaR-based ERC portfolios achieve positive Sharpe differences relative to equal weight across the universes.
\end{abstract}
\maketitle
\section{Introduction}
	Risk parity constructs portfolio weights from risk characteristics rather than expected-return maximization or direct portfolio-risk minimization. For example, inverse risk parity (IRP) assigns portfolio weights inversely proportional to standalone volatility, and equal risk contribution (ERC) equalizes volatility contributions across assets \citep{qian2011risk,maillard2010properties}. However, a limitation of IRP and ERC is that volatility treats upside and downside variation symmetrically. Skewness, heavy tails, and nonlinear dependence in asset returns are also not reflected in volatility-based portfolio construction.

	Several studies extend risk parity beyond volatility-based measures and Gaussian return models. Conditional Value-at-Risk (CVaR) deviation and mean absolute deviation have been considered in risk budgeting \citep{cesarone2018cvar_deviation,ararat2024mad}. \citet{choi2025diversified} develop diversified reward--risk parity rules based on standalone reward--risk measures including pure risk measures such as VaR and CVaR as well as standard deviation. Tail-risk-based portfolio construction has also been explored under tempered stable models. \citet{mercuri2014parametric} construct parametric risk parity portfolios using modified VaR and modified Expected Shortfall, and \citet{kim2022portfolio} derives marginal VaR and CVaR for a multivariate normal tempered stable (MNTS) model. EVaR has also been used in risk budgeting. \citet{cajas2021entropic} examines risk parity portfolios based on EVaR in a sample-based exponential-cone framework, and \citet{dacosta2023riskbudgeting} evaluate EVaR-based risk budgeting algorithms with simulated Gaussian and Student \(t\) returns.

	Tempered stable models have also been used for EVaR applications. \citet{nedeltchev2026measuring} examine EVaR market-risk measurement with an exponential tempered stable model. \citet{choi2026evarportfolio} develops portfolio-level EVaR optimization for MNTS and ICA-based tempered stable models and derives the corresponding weight-dependent moment-generating-function domains, but does not consider risk budgeting. Existing EVaR-based risk budgeting work does not use tempered stable portfolio models, while tempered stable risk parity research relies on other tail-risk measures. EVaR-based risk budgeting for tempered stable portfolio models remains unaddressed. Such an extension requires asset-level EVaR contributions for ERC construction.

	In this paper, we consider EVaR-based IRP and ERC portfolios with MNTS and ICA-based tempered stable models. First, we derive asset-level marginal EVaR and Euler contributions for both representations and compare them with the corresponding CVaR contributions. Second, we use the EVaR-deviation decomposition to separate the fitted location term from EVaR risk contributions. For Gaussian returns, the EVaR deviation is proportional to volatility, and the corresponding IRP and ERC rules recover conventional volatility IRP and ERC. Third, we evaluate the resulting portfolios using cross-asset ETFs, momentum-sorted portfolios, and U.S.\ sector ETFs. We compare the resulting portfolios with equal weight and a Gaussian EVaR benchmark, examine matched EVaR and CVaR portfolios with and without transaction costs, and compare raw EVaR-ERC with EVaR-deviation ERC portfolios.

	The remainder of this paper is organized as follows. Section~\ref{sec_rp_formulations} develops the EVaR-based risk budgeting framework. Section~\ref{sec_ts_computation} presents the tempered stable representations and the calculation of risk contributions. Section~\ref{sec_result} reports the empirical results. Section~\ref{sec_conclusion} concludes the paper.

\section{EVaR-based risk budgeting}
\label{sec_rp_formulations}
	This section develops the EVaR-based risk budgeting framework. We first review conventional volatility-based IRP and ERC, define their EVaR counterparts, and introduce the EVaR deviation.

	Throughout this paper, bold symbols denote vectors. Let \(\boldsymbol{X}_t=(X_t^{(1)},\ldots,X_t^{(N)})^{\top}\) be the vector of asset returns at horizon \(t\). For the one-period empirical analysis, we write \(\boldsymbol{X}=\boldsymbol{X}_1\).

	We consider long-only fully invested portfolios with weights \(\w\in\Delta_N\), where
	\begin{align}
		\Delta_N=\{\w\in\mathbb{R}^N:\wt\boldsymbol{e}=1,\;w_i\ge0\},
	\end{align}
	and \(\boldsymbol{e}=(1,\ldots,1)^{\top}\). For a given weight vector, the portfolio return and loss are given by
	\begin{align}
	\label{eq_sign_convention}
		&R(\w)=\wt\boldsymbol{X},\\
		&L(\w)=-\wt\boldsymbol{X}=-R(\w).
	\end{align}

\subsection{Volatility-based IRP and ERC}
	The conventional IRP portfolio uses inverse standalone volatility. Let \(\boldsymbol{\Omega}\) denote the return covariance matrix, \(\mathrm{Cov}(\boldsymbol{X})\). The weights of the IRP portfolio are given by
	\begin{align}
	\label{eq_irp_var}
		w_i^{\mathrm{IRP}}=\frac{\sigma_i^{-1}}{\sum_{j=1}^N\sigma_j^{-1}},
	\end{align}
	where \(\sigma_i=\sqrt{\Omega_{ii}}\). IRP assigns larger weights to assets with lower volatility.
	
	The volatility contribution of asset \(i\) is defined as
	\begin{align}
		RC_i^{\mathrm{Vol}}(\w)=w_i\frac{\partial\sigma(\w)}{\partial w_i}=\frac{w_i(\boldsymbol{\Omega}\w)_i}{\sqrt{\wt\boldsymbol{\Omega}\w}},
	\end{align}
	where \(\sigma(\w)\) is the portfolio volatility, i.e., \(\sigma(\w)=\sqrt{\wt\boldsymbol{\Omega}\w}\).
	
	Since ERC equalizes volatility contributions across all assets \citep{maillard2010properties}, its defining condition is given by
	\begin{align}
	\label{eq_erc_var}
		RC_i^{\mathrm{Vol}}(\w)=RC_j^{\mathrm{Vol}}(\w),
	\end{align}
	for all \(i\) and \(j\), subject to \(\w\in\Delta_N\).

\subsection{Entropic Value-at-Risk}
\label{sec_evar_prelim}
	Following~\citet{ahmadi2012entropic}, let \(\eta\in(0,1)\) be the tail probability and \(1-\eta\) the confidence level. For a loss random variable \(L\), EVaR is defined as
	\begin{align}
	\label{eq_evar_def}
		\mathrm{EVaR}_{1-\eta}(L)=\inf_{u\in\mathcal{U}}\frac{\log M_L(u)-\ln\eta}{u},
	\end{align}
	where \(M_L(u)\) is the loss moment-generating function, and \(\mathcal{U}\) contains the positive values of \(u\) for which \(M_L(u)\) is finite. 
	
	As in~\citet{choi2026evarportfolio}, the distributional parameters are estimated from returns. When \(X\) denotes a return, \(L=-X\). We use \(\mathrm{EVaR}_{1-\eta}(X)\) as shorthand for the EVaR of the corresponding loss:
	\begin{align}
	\label{eq_evar_return_convention}
		\mathrm{EVaR}_{1-\eta}(X):=\mathrm{EVaR}_{1-\eta}(L)=\inf_{u\in\mathcal{U}}\frac{\log M_L(u)-\ln\eta}{u}=\inf_{u\in\mathcal{U}}\frac{\log M_X(-u)-\ln\eta}{u}=\inf_{u\in\mathcal{U}}\frac{\Psi_X(-u)-\ln\eta}{u},
	\end{align}
	where \(\Psi_X(z)=\log M_X(z)\) denotes the scalar cumulant-generating function (CGF). For a random vector \(\boldsymbol{X}\), the corresponding CGF is given by
	\begin{align}
		\Psi_{\boldsymbol{X}}(\boldsymbol{z})=\log\mathbb E[\exp(\boldsymbol{z}^\top\boldsymbol{X})].
	\end{align}
	
	Applying the return-side convention to the portfolio return \(\wt\boldsymbol{X}\), the portfolio EVaR is expressed as
	\begin{align}
	\label{eq_rho_evar_portfolio}
	\mathcal{D}_{\eta}(\w)
	=
	\mathrm{EVaR}_{1-\eta}(\wt\boldsymbol{X}).
	\end{align}

\subsection{EVaR-based IRP}
	Following~\citet{choi2025diversified}, the inverse-risk construction can extend conventional IRP from volatility to other risk measures. We apply this construction to EVaR. For asset \(i\), the standalone EVaR is defined as
	\begin{align}
	\label{eq_standalone_evar}
		\mathcal{D}_{\eta,i}	=\mathrm{EVaR}_{1-\eta}(X^{(i)}).
	\end{align}
	When all \(\mathcal{D}_{\eta,i}\) are positive, the EVaR-IRP weights based on inverse-risk normalization are given by
	\begin{align}
	\label{eq_irp_evar}
		w_i^{\mathrm{EVaR\text{-}IRP}}=\frac{(\mathcal{D}_{\eta,i})^{-1}}{\sum_{j=1}^N(\mathcal{D}_{\eta,j})^{-1}}.
	\end{align}
	The EVaR-IRP rule assigns larger weights to assets with smaller EVaR, as conventional IRP does for volatility. It does not equalize portfolio risk contributions.

\subsection{EVaR-based ERC}
	Risk budgeting can be defined for a broad class of positively homogeneous risk measures \citep{cetingoz2024riskbudgeting}. \citet{cajas2021entropic} considers risk parity portfolios based on EVaR in a sample-based exponential-cone setting. We apply the Euler risk budgeting condition to the parametric EVaR functional \(\mathcal D_{\eta}\). For the Euler decomposition, we consider \(\mathcal D_{\eta}\) before restricting the portfolio weights to \(\Delta_N\), and positive homogeneity gives the following relation for any \(c>0\):
	\begin{align}
		\mathcal D_{\eta}(c\w)=c\mathcal D_{\eta}(\w).
	\end{align}
	If \(\mathcal D_{\eta}\) is differentiable at \(\w\), differentiating this relation with respect to \(c\) and evaluating it at \(c=1\) gives
	\begin{align}
	\label{eq_evar_euler}
		\mathcal{D}_{\eta}(\w)=\sum_{i=1}^N	w_i\frac{\partial\mathcal{D}_{\eta}(\w)}{\partial w_i}.
	\end{align}
	From this, we define the EVaR Euler contribution of asset \(i\) by
	\begin{align}
	\label{eq_evar_rc}
		\mathrm{EVaRC}_{\eta,i}(\w)=w_i\frac{\partial\mathcal{D}_{\eta}(\w)}{\partial w_i},
	\end{align}
	and the EVaR-ERC portfolio equalizes these contributions for any assets \(i\) and \(j\):
	\begin{align}
	\label{eq_erc_evar}
		\mathrm{EVaRC}_{\eta,i}(\w)=\mathrm{EVaRC}_{\eta,j}(\w).
	\end{align}
	Using the Euler decomposition, the same condition assigns an equal share of total portfolio EVaR to each asset:
	\begin{align}
	\label{eq_evar_erc_equal}
		\mathrm{EVaRC}_{\eta,i}(\w)=\frac{1}{N}\mathcal{D}_{\eta}(\w).
	\end{align}
	This differs from minimum-EVaR optimization. ERC allocates portfolio EVaR across assets instead of minimizing total EVaR.

\subsection{EVaR deviation and location-free risk budgeting}
\label{sec_evar_deviation}
	Since EVaR is translation equivariant, its Euler contributions include the fitted location term. Centered coherent risk measures generate generalized deviations \citep{rockafellar2006generalized}. \citet{grechuk2026riskquadrangle} include the EVaR deviation in the divergence-based risk-quadrangle framework. We use this deviation to separate the fitted location term from the remaining distributional risk. Let \(\boldsymbol{\mu}=\mathbb E[\boldsymbol{X}]\) be the fitted mean and define \(\boldsymbol{X}^c=\boldsymbol{X}-\boldsymbol{\mu}\). The EVaR deviation is defined as
	\begin{align}
	\label{eq_evar_deviation}
		\mathcal D_{\eta}^{c}(\w)=	\mathrm{EVaR}_{1-\eta}\!\left(\wt\boldsymbol{X}^c\right).
	\end{align}
	Centering removes the fitted location term but leaves the remaining distributional structure unchanged.
	
	When the EVaR deviation is differentiable, we denote its Euler contribution for asset \(i\) by \(\mathrm{DEVaRC}_{\eta,i}\):
	\begin{align}
	\label{eq_evar_deviation_rc}
		\mathrm{DEVaRC}_{\eta,i}(\w)=w_i\frac{\partial\mathcal D_{\eta}^{c}(\w)}{\partial w_i}.
	\end{align}

	The coherent-risk/deviation correspondence above gives the following EVaR-specific location separation.

	\begin{prp}[EVaR location separation]
		\label{prp_location_separation}
		If the loss-side moment-generating-function domain is nonempty, the raw EVaR and EVaR-deviation functionals satisfy
		\begin{align}
		\label{eq_evar_deviation_translation}
			\mathcal D_{\eta}(\w)=	-\wt\boldsymbol{\mu}+\mathcal D_{\eta}^{c}(\w).
		\end{align}
		At weights where both functionals are differentiable, their Euler contributions satisfy
		\begin{align}
		\label{eq_rc_location_decomposition}
			\mathrm{EVaRC}_{\eta,i}(\w)=-w_i\mu_i+\mathrm{DEVaRC}_{\eta,i}(\w).
		\end{align}
	\end{prp}
	\begin{proof}
		The raw and centered CGFs differ only by the location term at the loss-side argument:
		\begin{align*}
			\Psi_{\boldsymbol{X}}(-u\w)	=-u\wt\boldsymbol{\mu}+\Psi_{\boldsymbol{X}^c}(-u\w).
		\end{align*}
		Since centering does not change the moment-generating-function domain, substituting the CGF decomposition into the EVaR objective gives
		\begin{align*}
			\frac{\Psi_{\boldsymbol{X}}(-u\w)-\ln\eta}{u}=-\wt\boldsymbol{\mu}+\frac{\Psi_{\boldsymbol{X}^c}(-u\w)-\ln\eta}{u}.
		\end{align*}
		Since the first term on the right-hand side does not depend on \(u\), the raw and centered objectives have the same minimizing \(u\). Eq.~\eqref{eq_evar_deviation_translation} follows directly. At a differentiable weight, differentiating with respect to \(w_i\) and multiplying by \(w_i\) gives Eq.~\eqref{eq_rc_location_decomposition}. 
	\end{proof}

	Applying the equal-contribution condition to the EVaR deviation, the EVaR-deviation risk contribution for each asset \(i\) satisfies
	\begin{align}
	\label{eq_evar_deviation_erc}
		\mathrm{DEVaRC}_{\eta,i}(\w)=\frac{1}{N}\mathcal D_{\eta}^{c}(\w).
	\end{align}

	\begin{crl}[Gaussian benchmark]
		\label{cor_gaussian_benchmark}
		Under the Gaussian benchmark \(\boldsymbol{X}\sim N(\boldsymbol{\mu},\boldsymbol{\Omega})\), the EVaR deviation is given by
		\begin{align}
		\label{eq_gaussian_evar_deviation}
			\mathcal D_{\eta}^{c}(\w)=k_\eta\sqrt{\wt\boldsymbol{\Omega}\w},
		\end{align}
		where \(k_\eta=\sqrt{-2\ln\eta}\). Its Euler contribution is given by
		\begin{align}
		\label{eq_gaussian_deviation_rc}
			\mathrm{DEVaRC}_{\eta,i}(\w)=k_\eta\frac{w_i(\boldsymbol{\Omega}\w)_i}{\sqrt{\wt\boldsymbol{\Omega}\w}}.
		\end{align}
		EVaR-deviation IRP and ERC give the same portfolio weights as conventional volatility IRP and ERC.
	\end{crl}
	\begin{proof}
		For centered Gaussian returns, \(\Psi_{\boldsymbol{X}^c}(-u\w)=u^2\wt\boldsymbol{\Omega}\w/2\). Substituting this expression into the EVaR-deviation objective and minimizing over \(u>0\) gives Eq.~\eqref{eq_gaussian_evar_deviation}. For a single asset, \(\mathcal D_{\eta,i}^{c}=k_\eta\sigma_i\). Inverse-risk normalization gives the conventional volatility IRP weights. Differentiating Eq.~\eqref{eq_gaussian_evar_deviation} gives Eq.~\eqref{eq_gaussian_deviation_rc}. The common factor \(k_\eta\) leaves the equal-contribution conditions unchanged, giving the conventional volatility ERC weights.
	\end{proof}

\section{Tempered stable representations and risk-contribution computation}
\label{sec_ts_computation}
	In this section, we extend the EVaR-based risk budgeting framework developed in Section~\ref{sec_rp_formulations} to the MNTS and ICA-based tempered stable models and derive the corresponding asset-level EVaR contributions. We use the portfolio representations, projected cumulants, and weight-dependent moment-generating-function domains developed in \citet{choi2026evarportfolio} as inputs to the risk-contribution analysis.

	Following \citet{choi2026evarportfolio}, let \(\mathcal A\in\{\mathrm{MNTS},\mathrm{ICA}\}\) denote the EVaR evaluation approach. We use \(\mathcal D_{\eta}^{\mathcal A}\) for portfolio EVaR under approach \(\mathcal A\), \(\mathrm{EVaRC}_{\eta,i}^{\mathcal A}\) for the corresponding EVaR risk contribution, and \(\mathrm{DEVaRC}_{\eta,i}^{\mathcal A}\) for the corresponding EVaR-deviation risk contribution. In the ICA approach, \(\mathcal F\) denotes either the normal tempered stable (NTS) or classical tempered stable (CTS) component family used in \citet{choi2026evarportfolio}.

\subsection{Portfolio representations and admissible domains}
\label{sec_distributional_routes}
	We summarize the MNTS and ICA portfolio representations and admissible domains used below.

\paragraph{MNTS approach}
	For the MNTS specification, the portfolio parameterization and its weight-dependent admissible domain follow~\citet{choi2026evarportfolio}. Let \((\alpha,\theta,\boldsymbol{\beta},\boldsymbol{\gamma},\boldsymbol{\mu},\boldsymbol{\rho})\) denote the fitted parameter set, where \(\alpha\) and \(\theta\) are the common NTS shape parameters, \(\boldsymbol{\beta}\) is the time-change loading vector, \(\boldsymbol{\gamma}\) is the diffusion-scale vector, \(\boldsymbol{\mu}\) is the location vector, and \(\boldsymbol{\rho}\) is the latent Brownian correlation matrix. The diffusion-scale and correlation parameters determine the diffusion covariance matrix
	\begin{align}
		\boldsymbol{\Sigma}=\mathrm{diag}(\boldsymbol{\gamma})\boldsymbol{\rho}\,\mathrm{diag}(\boldsymbol{\gamma}).
	\end{align}
	Since the common time-change term adds to the diffusion covariance, these parameters imply the one-period return covariance:
	\begin{align}
		\boldsymbol{\Omega}=\mathrm{Cov}(\boldsymbol{X})=\boldsymbol{\Sigma}+\frac{2-\alpha}{2\theta}\boldsymbol{\beta}\boldsymbol{\beta}^\top.
	\end{align}
	It also includes the covariance induced by the common time change through \(\boldsymbol{\beta}\). The two matrices \(\boldsymbol{\Omega}\) and \(\boldsymbol{\Sigma}\) coincide when \(\boldsymbol{\beta}=\boldsymbol{0}\).

	Since the MNTS distribution is closed under linear portfolio aggregation, a portfolio with weights \(\w\) remains NTS with parameters:
	\begin{align}
	\label{eq_mnts_projected_params}
		\bar\alpha=\alpha,\qquad
		\bar\theta=\theta,\qquad
		\bar\beta=\wt\boldsymbol{\beta},\qquad
		\bar\gamma=\sqrt{\wt\boldsymbol{\Sigma}\w},\qquad
		\bar\mu=\wt\boldsymbol{\mu}.
	\end{align}
	With these parameters, the portfolio CGF and its weight-dependent loss-side domain are given by
	\begin{align}
	\label{eq_mnts_cgf}
		\Psi_{\boldsymbol{X}}(-u\w)&=-u(\bar\mu-\bar\beta)	-\frac{2\theta}{\alpha}\left[g(u)^{\alpha/2}-1\right],\\
	\label{eq_mnts_domain_rp}
		\mathcal U_{\mathrm{MNTS}}(\w)&=\left(0,u_{\max}^{\mathrm{MNTS}}(\w)\right],\\
		u_{\max}^{\mathrm{MNTS}}(\w)&=\frac{\bar\beta+\sqrt{\bar\beta^2+2\theta\bar\gamma^2}}{\bar\gamma^2},
	\end{align}
	where \(g(u)\) is the quadratic term written as
	\begin{align}
	\label{eq_mnts_g}
		g(u)=1+\frac{\bar\beta}{\theta}u-\frac{\bar\gamma^2}{2\theta}u^2.
	\end{align}
	These quantities determine portfolio EVaR and the asset-level contributions below. For the EVaR deviation with MNTS distribution, \(\boldsymbol{\mu}\) is set to zero while \((\alpha,\theta,\boldsymbol{\beta},\boldsymbol{\gamma},\boldsymbol{\rho})\) remain unchanged.

\paragraph{ICA approach}
	For ICA, we use the component representation and admissible-domain characterization in~\citet{choi2026evarportfolio}. The one-period return representation and portfolio component exposures are given by
	\begin{align}
	\label{eq_ica_decomp}
		\boldsymbol{X}=\boldsymbol m+\mathbf A\boldsymbol S,\qquad \tw=\mathbf A^\top\w,
	\end{align}
	where \(\boldsymbol m\) is the sample centering vector used before ICA, \(\mathbf A\) is the mixing matrix, and the components are fitted independently under the chosen family \(\mathcal F\). The vector \(\boldsymbol m\) need not equal the fitted joint mean because the fitted component laws can have nonzero means. Independence of the fitted components gives the following portfolio CGF and admissible domain:
	\begin{align}
	\label{eq_ica_cgf}
		\Psi_{\boldsymbol{X}}(-u\w)&=-u\wt\boldsymbol m+\sum_{j=1}^{J}\Psi_{S^{(j)}}(-u\tilde w_j),\\
	\label{eq_ica_domain_rp}
		\mathcal U_{\mathrm{ICA}}(\w)&=\bigcap_{j=1}^{J}\{u>0:-u\tilde w_j\in\mathrm{dom}(M_{S^{(j)}})\}.
	\end{align}
	In the risk-contribution formulas below, \(i\) indexes assets and \(j\) indexes independent components.

	As shown in~\citet{choi2026evarportfolio}, for the NTS and CTS component specifications used here, Eq.~\eqref{eq_ica_domain_rp} can be written as
	\begin{align}
	\label{eq_ica_domain_endpoint_rp}
		\mathcal U_{\mathrm{ICA}}(\w)=\left(0,\min_{1\le j\le J}u_{j,\max}(\tilde w_j)\right],
	\end{align}
	where the component-specific NTS and CTS formulas for \(u_{j,\max}\) are given in~\citet{choi2026evarportfolio}. Only their dependence on \(\tw=\mathbf A^\top\w\) is needed here.

	For the EVaR-deviation formulation, the fitted component means and the fitted joint mean \(\boldsymbol{\mu}\) introduced in Section~\ref{sec_evar_deviation} are given by
	\begin{align}
	\label{eq_ica_component_mean}
		\boldsymbol{\mu}_S&=(\mu_{S,1},\ldots,\mu_{S,J})^\top,\qquad	\mu_{S,j}=\mathbb E[S^{(j)}]=\Psi'_{S^{(j)}}(0),\\
	\label{eq_ica_fitted_mean}
		\boldsymbol{\mu}&=\boldsymbol m+\mathbf A\boldsymbol{\mu}_S.
	\end{align}
	The centered return representation takes the form
	\begin{align}
	\label{eq_ica_centered_representation}
		\boldsymbol{X}^c=\mathbf A(\boldsymbol S-\boldsymbol{\mu}_S).
	\end{align}
	For each component, the centered CGF is given by
	\begin{align}
	\label{eq_ica_centered_component_cgf}
		\Psi^c_{S^{(j)}}(z)=\Psi_{S^{(j)}}(z)-z\mu_{S,j}.
	\end{align}

\subsection{EVaR risk contribution computation}
\label{sec_rc_computation}
	We first derive analytic EVaR risk contributions and present a finite-difference formulation that also applies when the EVaR optimum lies on a weight-dependent admissible-domain boundary.

\subsubsection{Risk contributions from analytic derivatives}
\label{sec_analytic_rc}
	Let \(u^\star(\w)\) be a unique minimizer of the scalar EVaR problem and suppose that it lies in the interior of the admissible domain. If the portfolio CGF is differentiable at \(-u^\star(\w)\w\), the envelope theorem gives
	\begin{align}
	\label{eq_envelope}
		\frac{\partial\mathcal D_{\eta}^{\mathcal A}(\w)}{\partial w_i}=-\frac{\partial\Psi_{\boldsymbol{X}}(\boldsymbol z)}{\partial z_i}\bigg|_{\boldsymbol z=-u^\star(\w)\w}.
	\end{align}
	\begin{prp}[MNTS EVaR risk contributions]
	Suppose that the MNTS EVaR problem has a unique interior minimizer \(u^\star(\w)\), and let \(g^\star=g(u^\star(\w))\). The asset-level EVaR Euler contribution is given by
	\begin{align}
	\label{eq_rc_mnts}
		\mathrm{EVaRC}_{\eta,i}^{\mathrm{MNTS}}(\w)=w_i\Bigg[-\mu_i+\beta_i\{1-(g^\star)^{\alpha/2-1}\}+u^\star(\w)(\boldsymbol{\Sigma}\w)_i(g^\star)^{\alpha/2-1}\Bigg].
	\end{align}
	The corresponding EVaR-deviation contribution is given by
	\begin{align}
	\label{eq_rc_mnts_centered}
		\mathrm{DEVaRC}_{\eta,i}^{\mathrm{MNTS}}(\w)=\mathrm{EVaRC}_{\eta,i}^{\mathrm{MNTS}}(\w)+w_i\mu_i.
	\end{align}
	\end{prp}
	\begin{proof}
	From Eqs.~\eqref{eq_mnts_projected_params}--\eqref{eq_mnts_cgf}, differentiation of \(\Psi_{\boldsymbol{X}}(-u\w)\) with respect to \(w_i\) gives
	\begin{align*}
		\frac{1}{u}\frac{\partial\Psi_{\boldsymbol{X}}(-u\w)}{\partial w_i}=-\mu_i+\beta_i\{1-g(u)^{\alpha/2-1}\}+u(\boldsymbol{\Sigma}\w)_i g(u)^{\alpha/2-1}.
	\end{align*}
	Evaluating at \(u=u^\star(\w)\) and multiplying by \(w_i\) gives Eq.~\eqref{eq_rc_mnts}. Eq.~\eqref{eq_rc_mnts_centered} follows from Proposition~\ref{prp_location_separation} and Eq.~\eqref{eq_rc_location_decomposition}.
	\end{proof}

	\begin{prp}[ICA EVaR risk contributions]
	Suppose that the ICA EVaR problem has a unique interior minimizer \(u^\star(\w)\). The asset-level EVaR Euler contribution is given by
	\begin{align}
	\label{eq_rc_ica}
		\mathrm{EVaRC}_{\eta,i}^{\mathrm{ICA}}(\w)=w_i\left[-m_i-\sum_{j=1}^{J}A_{ij}\Psi_{S^{(j)}}'(-u^\star(\w)\tilde w_j)\right].
	\end{align}
	The corresponding EVaR-deviation contribution is given by
	\begin{align}
		\mathrm{DEVaRC}_{\eta,i}^{\mathrm{ICA}}(\w)
		&=\mathrm{EVaRC}_{\eta,i}^{\mathrm{ICA}}(\w)+w_i\mu_i\notag\\
		&=-w_i\sum_{j=1}^J A_{ij}\left[\Psi_{S^{(j)}}'(-u^\star(\w)\tilde w_j)-\mu_{S,j}\right]\notag\\
	\label{eq_rc_ica_deviation}
		&=-w_i\sum_{j=1}^J A_{ij}(\Psi^c_{S^{(j)}})'(-u^\star(\w)\tilde w_j).
	\end{align}
	\end{prp}
	\begin{proof}
	Differentiating Eq.~\eqref{eq_ica_cgf} with respect to \(w_i\) and using \(\partial\tilde w_j/\partial w_i=A_{ij}\) gives
	\begin{align}
	\label{eq_ica_grad}
		\frac{1}{u}\frac{\partial\Psi_{\boldsymbol{X}}(-u\w)}{\partial w_i}=-m_i-\sum_{j=1}^{J}A_{ij}\Psi_{S^{(j)}}'(-u\tilde w_j).
	\end{align}
	Evaluating at \(u=u^\star(\w)\) and multiplying by \(w_i\) gives Eq.~\eqref{eq_rc_ica}. Eq.~\eqref{eq_rc_ica_deviation} follows from Eqs.~\eqref{eq_ica_centered_component_cgf} and~\eqref{eq_rc_location_decomposition}.
	\end{proof}

	The MNTS contribution depends on the common time-change loading \(\boldsymbol{\beta}\), diffusion covariance \(\boldsymbol{\Sigma}\), and the tilted factor \((g^\star)^{\alpha/2-1}\). The ICA counterpart depends on the mixing matrix and the component CGF derivatives evaluated at the portfolio-specific exponential tilt.

\subsubsection{Risk contributions from finite differences}
\label{sec_fd_implementation}
	The analytic contributions above require a unique interior EVaR minimizer. A finite-difference formulation can also be used when this interiority condition does not hold.

	We calculate the coordinate derivative of the fitted EVaR functional using a central difference:
	\begin{align}
	\label{eq_fd_grad}
		\partial_i\mathcal D_{\eta}^{\mathcal A}(\w)=\frac{\mathcal D_{\eta}^{\mathcal A}(\w+h e_i)-\mathcal D_{\eta}^{\mathcal A}(\w-h e_i)}{2h},
	\end{align}
	where EVaR and its admissible moment-generating-function domain are recalculated for each perturbed portfolio. In particular, this captures changes in the MNTS endpoint in Eq.~\eqref{eq_mnts_domain_rp} and the ICA endpoint in Eq.~\eqref{eq_ica_domain_endpoint_rp}. Multiplying the finite-difference derivative by the portfolio weight gives the numerical Euler contribution
	\begin{align}
	\label{eq_fd_rc}
		\mathrm{EVaRC}_{\eta,i}^{\mathcal A}(\w)=w_i\partial_i\mathcal D_{\eta}^{\mathcal A}(\w).
	\end{align}
	For the EVaR deviation, we apply the same finite-difference formulas after replacing \(\mathcal D_{\eta}^{\mathcal A}\) and \(\mathrm{EVaRC}_{\eta,i}^{\mathcal A}\) with \(\mathcal D_{\eta}^{c,\mathcal A}\) and \(\mathrm{DEVaRC}_{\eta,i}^{\mathcal A}\). Since the perturbations approximate unconstrained coordinate derivatives at a feasible portfolio, the perturbed weights need not satisfy the budget constraint.
	
	At an active endpoint, a weight perturbation also changes the admissible endpoint. Recomputing EVaR at \(\w\pm h e_i\) captures this change except at nonsmooth switching points.

\subsection{EVaR--CVaR contribution comparison}
\label{sec_evar_cvar_contribution_comparison}
	We compare EVaR and CVaR risk contributions under the same fitted return distribution. We denote the CVaR risk contribution by \(\mathrm{CVaRC}_{\eta,i}^{\mathcal A}\) and the corresponding CVaR-deviation risk contribution by \(\mathrm{DCVaRC}_{\eta,i}^{\mathcal A}\).

\paragraph{MNTS approach}
	For a continuous return \(R\), return-side CVaR is defined as \(-\mathbb E[R\mid R\le q_\eta(R)]\), where \(q_\eta(R)\) is the lower-tail \(\eta\)-quantile. \citet{kim2022portfolio} derives closed-form marginal VaR and CVaR contributions for the multivariate NTS market model. We write that CVaR result in the notation used here and define
	\begin{align}
	&s(\w)=\sqrt{\wt\boldsymbol{\Omega}\w},\\
	&b(\w)=\frac{\wt\boldsymbol{\beta}}{s(\w)},
	\end{align}
	and let \(C_\eta(b)\) denote the return-side CVaR of the corresponding mean-zero, unit-variance standardized NTS variable for fixed shape parameters \((\alpha,\theta)\) and standardized asymmetry coordinate \(b\). Combining the portfolio scale \(s(\w)\) and standardized asymmetry \(b(\w)\) with this tail functional gives portfolio CVaR as
	\begin{align}
	\label{eq_cvar_mnts_std}
		\operatorname{CVaR}_{1-\eta}^{\mathrm{MNTS}}(\w)=-\wt\boldsymbol{\mu}+s(\w)C_\eta\!\left(b(\w)\right).
	\end{align}
	Differentiating the standardized MNTS representation with respect to the weights gives the Euler contribution when \(C_\eta(b)\) is differentiable:
	\begin{align}
	\label{eq_rc_cvar_mnts}
		\mathrm{CVaRC}_{\eta,i}^{\mathrm{MNTS}}(\w)=w_i\Bigg[-\mu_i+C_\eta(b)\frac{(\boldsymbol{\Omega}\w)_i}{s}+C_\eta'(b)\left\{\beta_i-b\frac{(\boldsymbol{\Omega}\w)_i}{s}\right\}
	\Bigg],
	\end{align}
	where \(s=s(\w)\) and \(b=b(\w)\). Removing the same fitted location term from the MNTS CVaR contribution gives the CVaR-deviation contribution
	\begin{align}
	\label{eq_rc_cvar_mnts_centered}
		\mathrm{DCVaRC}_{\eta,i}^{\mathrm{MNTS}}(\w)=\mathrm{CVaRC}_{\eta,i}^{\mathrm{MNTS}}(\w)+w_i\mu_i.
	\end{align}
	The common fitted location term cancels from the EVaR--CVaR contribution difference:
	\begin{align}
	\label{eq_mnts_centered_evar_cvar_difference}
		\mathrm{DEVaRC}_{\eta,i}^{\mathrm{MNTS}}(\w)-\mathrm{DCVaRC}_{\eta,i}^{\mathrm{MNTS}}(\w)
		=\mathrm{EVaRC}_{\eta,i}^{\mathrm{MNTS}}(\w)-\mathrm{CVaRC}_{\eta,i}^{\mathrm{MNTS}}(\w).
	\end{align}
	The remaining tail-allocation terms differ. EVaR uses the exponentially tilted CGF in Eq.~\eqref{eq_rc_mnts}, while CVaR uses the standardized NTS tail functional and its asymmetry derivative in Eq.~\eqref{eq_rc_cvar_mnts}.

\paragraph{ICA approach}
	For a continuous portfolio-return distribution, the return-side CVaR Euler derivative equals minus the lower-tail conditional mean of the asset return \citep{rockafellar2000optimization,rockafellar2002conditional}. Define the lower-tail quantile of the ICA source portfolio by
	\begin{align}
	\label{eq_ica_source_quantile}
		q_{\eta}^{S}(\w)=\inf\left\{x:\Pr\left(\sum_{k=1}^{J}\tilde w_k S^{(k)}\le x\right)\ge\eta\right\}.
	\end{align}
	Conditioning each component on this portfolio tail event gives the corresponding component tail means:
	\begin{align}
	\label{eq_ica_cvar_tail_mean}
		\tau_{\eta,j}(\w)=\mathbb E\left[S^{(j)}\,\middle|\,\sum_{k=1}^{J}\tilde w_k S^{(k)}\le q_{\eta}^{S}(\w)\right].
	\end{align}
	These tail quantities determine portfolio CVaR and the corresponding asset-level Euler contributions:
	\begin{align}
	\label{eq_cvar_ica}
		\operatorname{CVaR}_{1-\eta}^{\mathrm{ICA}}(\w)=-\wt\boldsymbol m-\sum_{j=1}^{J}\tilde w_j\tau_{\eta,j}(\w),\\
	\label{eq_rc_cvar_ica}
		\mathrm{CVaRC}_{\eta,i}^{\mathrm{ICA}}(\w)=w_i\left[-m_i-\sum_{j=1}^{J}A_{ij}\tau_{\eta,j}(\w)\right].
	\end{align}
	The ICA components are independent, but the tail event in Eq.~\eqref{eq_ica_cvar_tail_mean} depends on their weighted sum. Eqs.~\eqref{eq_cvar_ica}--\eqref{eq_rc_cvar_ica} are exact tail-conditional formulas, but they do not reduce to a scalar marginal formula for general ICA+NTS or ICA+CTS specifications. Since centering changes only the fitted mean, the ICA CVaR-deviation contribution becomes
	\begin{align}
		\mathrm{DCVaRC}_{\eta,i}^{\mathrm{ICA}}(\w)&=\mathrm{CVaRC}_{\eta,i}^{\mathrm{ICA}}(\w)+w_i\mu_i\notag\\
	\label{eq_rc_cvar_ica_centered}
		&=-w_i\sum_{j=1}^{J}A_{ij}\{\tau_{\eta,j}(\w)-\mu_{S,j}\}.
	\end{align}
	
	For ICA, the distinction between EVaR and CVaR can be expressed through the EVaR tilted component mean
	\begin{align}
	\label{eq_ica_evar_tilted_mean}
		\xi_{\eta,j}^{\mathrm{EVaR}}(\w)=\Psi_{S^{(j)}}'(-u^\star(\w)\tilde w_j).
	\end{align}
	In Eqs.~\eqref{eq_rc_ica} and~\eqref{eq_rc_cvar_ica}, the two contributions use the same mixing matrix but different component quantities. EVaR uses an Esscher-type exponentially tilted component mean, while CVaR uses the component mean conditional on the portfolio lower-tail event. For the corresponding deviations, the common fitted location term cancels and the contribution difference reduces to
	\begin{align}
	\label{eq_ica_centered_evar_cvar_difference}
		\mathrm{DEVaRC}_{\eta,i}^{\mathrm{ICA}}(\w)-\mathrm{DCVaRC}_{\eta,i}^{\mathrm{ICA}}(\w)=-w_i\sum_{j=1}^{J}A_{ij}\left\{\xi_{\eta,j}^{\mathrm{EVaR}}(\w)-\tau_{\eta,j}(\w)\right\}.
	\end{align}
	Eqs.~\eqref{eq_mnts_centered_evar_cvar_difference} and~\eqref{eq_ica_centered_evar_cvar_difference} show that centering removes the common fitted location term in both representations. For MNTS, the remaining difference is between EVaR and CVaR contributions evaluated from the same projected NTS distribution. For ICA, it is driven by the difference between exponentially tilted and lower-tail conditional component means. This distinction provides the structural interpretation for the matched EVaR--CVaR comparisons in Section~\ref{sec_result}.

\section{Results}
\label{sec_result}
	We evaluate the resulting risk parity portfolios in three investment universes. The empirical analysis covers performance results relative to the equal weight portfolio and the Gaussian EVaR benchmark, matched EVaR--CVaR comparisons, transaction costs and turnover, and the effect of removing the fitted location term from EVaR-ERC.

\subsection{Dataset and methodology}
	This subsection describes the three investment universes, the portfolio construction, and the performance evaluation used in the backtests.
\subsubsection{Dataset for backtesting}
\paragraph{Cross-asset ETF universe}
	The first universe, denoted XASSET, consists of seven broad ETFs: U.S.\ equities (VTI), international developed equities (EFA), emerging-market equities (VWO), real estate (VNQ), aggregate U.S.\ bonds (AGG), inflation-protected Treasury bonds (TIP), and gold (GLD). The universe includes asset classes with different volatility, correlation, and tail characteristics. We obtain adjusted daily prices and returns from Yahoo Finance from April~1,~2005 to March~31,~2026. After the initial twelve-month estimation window, the out-of-sample period runs from April~2006 to March~2026.

\paragraph{Momentum-decile universe}
	The second universe, denoted MOM10, consists of ten momentum-sorted portfolios based on prior 12--2 returns~\citep{jegadeesh1993returns}. We obtain the value-weighted decile portfolio returns from the Kenneth R. French Data Library at Dartmouth College. Although the French momentum series extends further back, we use the portfolio returns from January~1995 to March~2026 for the empirical test in this paper. Since the first twelve months are used for the initial estimation window, the out-of-sample period begins in January~1996. Since the decile portfolios share a common equity component and a cross-sectional momentum structure, their dependence structure is different from that of XASSET.

\paragraph{SPDR sector ETF universe}
	The third universe, denoted SECTOR, consists of the State Street Select Sector SPDR ETFs. The original universe contains XLB, XLE, XLF, XLI, XLK, XLP, XLU, XLV, and XLY based on U.S.\ S\&P 500 equity market exposures. XLRE and XLC enter in October~2015 and in June~2018, respectively. We obtain adjusted daily prices and returns from Yahoo Finance. The investable universe varies from nine to eleven ETFs over the out-of-sample period from January~2000 to March~2026. At each rebalance date, we use only ETFs with full price history for the corresponding lookback window.

\subsubsection{Portfolio construction}
	The three universes use the same portfolio construction methodology. At each monthly rebalance date, we use the previous twelve months of daily returns to estimate the model parameters and hold the resulting portfolio for one month. All portfolios are long-only and fully invested. We set the tail probability to \(\eta=0.05\), corresponding to a 95\% confidence level for EVaR and CVaR. For the empirical EVaR-ERC and matched CVaR-ERC portfolios, risk contributions are computed using the finite-difference procedure in Section~\ref{sec_fd_implementation}. The fitted parameters are held fixed across weight perturbations, and each perturbation requires a new portfolio risk evaluation but not a new distribution fit. Portfolio CVaR is evaluated as in~\citet{choi2026evarportfolio}. Table~\ref{tab:pair_design_all_lb12_h1_rp} summarizes the matched design. Normal EVaR is reported as a benchmark. For EVaR-deviation ERC, Corollary~\ref{cor_gaussian_benchmark} shows that the Gaussian equal-contribution condition gives the same portfolio weights as conventional volatility ERC. MNTS, ICA+NTS, and ICA+CTS are evaluated with both EVaR and CVaR under IRP and ERC. We include IRP as the inverse-risk benchmark and ERC as the contribution-based risk budgeting portfolio.

\begingroup
\begin{table}[!t]
\centering
\caption{Experimental design of the risk-parity constructions. E denotes the EVaR-based variant and C the matched CVaR-based variant. The Normal specification is evaluated only for the entropic variant. }
\label{tab:pair_design_all_lb12_h1_rp}
\begin{tabular}{lcccc}
\toprule
 & Normal & MNTS projection & ICA$+$NTS & ICA$+$CTS \\
\midrule
IRP & E & E $+$ C & E $+$ C & E $+$ C \\
ERC & E & E $+$ C & E $+$ C & E $+$ C \\
\bottomrule
\end{tabular}
\end{table}
\endgroup

\subsubsection{Performance evaluation}
	We report cumulative return, annualized volatility, Sharpe ratio, maximum drawdown, turnover, and the additional performance measures used in the underlying backtests. Turnover is drift-adjusted one-way turnover per rebalance, and transaction costs are applied proportionally.

	For the Sharpe ratio comparisons, the null hypothesis is equality of the two population Sharpe ratios. We calculate two-sided \(p\)-values with a studentized circular block bootstrap following~\citet[Section 3]{ledoit2008robust}, using \(B=4{,}999\) resamples and block length \(L=\lfloor T^{1/3}\rfloor\). The tests are pairwise and are not adjusted for multiple testing.

\subsection{Benchmark comparisons}
\label{sec_equal_weight_results}
	We first assess the entropic risk parity portfolios against the equal weight portfolio and the Gaussian EVaR benchmark. Table~\ref{tab:benchmark_diff_sig_all_lb12_h1_rp} compares the entropic risk parity portfolios with the equal weight portfolio. These comparisons do not isolate the effect of EVaR because both portfolio construction and the risk measure differ from the equal-weight benchmark. ERC achieves a positive Sharpe difference in all nine market--specification comparisons. IRP obtains a positive difference in eight of the nine comparisons. The only negative IRP difference is ICA+CTS in XASSET. Nine of the 18 Sharpe differences have \(p\)-values below 0.10, including four below 0.05.

\begingroup
\begin{table}[!t]
\centering
\caption{Entropic risk-parity portfolios versus the equal-weight benchmark. The table reports the difference (portfolio minus benchmark) in CAGR and annualised volatility (percentage points), Sharpe ratio, and Calmar ratio. $^{*}$, $^{**}$, and $^{***}$ denote significance at the 10\%, 5\%, and 1\% levels. }
\label{tab:benchmark_diff_sig_all_lb12_h1_rp}
\small
\setlength{\tabcolsep}{5pt}
\begin{tabular}{lllrrrrr}
\toprule
 &  & & \multicolumn{5}{c}{vs EW} \\
\cmidrule(lr){4-8}
Market & Construction & Approach & $\Delta$CAGR & $\Delta$Vol & $\Delta$SR & $p$ & $\Delta$Calmar \\
\midrule
SECTOR & IRP & MNTS & -0.03 & -0.99 & +0.019 & 0.200 & +0.010 \\
 &  & ICA$+$NTS & +0.05 & -0.90 & +0.021 & 0.216 & +0.012 \\
 &  & ICA$+$CTS & +0.02 & -0.92 & +0.020 & 0.290 & +0.014 \\
\addlinespace
 & ERC & MNTS & -0.05 & -1.14 & +0.022 & 0.209 & +0.009 \\
 &  & ICA$+$NTS & +0.55 & -1.00 & +0.051$^{**}$ & 0.023 & +0.027 \\
 &  & ICA$+$CTS & +0.41 & -0.96 & +0.042$^{*}$ & 0.059 & +0.025 \\
\addlinespace
MOM10 & IRP & MNTS & +0.00 & -0.82 & +0.015 & 0.238 & +0.007 \\
 &  & ICA$+$NTS & +0.12 & -0.83 & +0.021 & 0.136 & +0.010 \\
 &  & ICA$+$CTS & +0.24 & -0.83 & +0.026$^{*}$ & 0.076 & +0.011 \\
\addlinespace
 & ERC & MNTS & -0.16 & -0.78 & +0.007 & 0.552 & +0.004 \\
 &  & ICA$+$NTS & +0.28 & -0.79 & +0.028$^{*}$ & 0.052 & +0.013 \\
 &  & ICA$+$CTS & +0.39 & -0.79 & +0.032$^{**}$ & 0.035 & +0.014 \\
\addlinespace
XASSET & IRP & MNTS & -1.50 & -5.68 & +0.177$^{**}$ & 0.016 & +0.048 \\
 &  & ICA$+$NTS & -1.41 & -5.16 & +0.142$^{*}$ & 0.057 & +0.043 \\
 &  & ICA$+$CTS & -2.27 & -2.95 & -0.077 & 0.710 & -0.074 \\
\addlinespace
 & ERC & MNTS & -1.52 & -6.37 & +0.246$^{**}$ & 0.011 & +0.089 \\
 &  & ICA$+$NTS & -1.45 & -5.53 & +0.170$^{*}$ & 0.086 & +0.076 \\
 &  & ICA$+$CTS & -1.52 & -4.77 & +0.100 & 0.511 & -0.017 \\
\bottomrule
\end{tabular}
\end{table}
\endgroup

	XASSET achieves the largest Sharpe improvements over equal weight. MNTS IRP and ERC exceed equal weight by 0.177 and 0.246 in Sharpe ratio, with \(p=0.016\) and \(p=0.011\). ICA+NTS IRP and ERC also exhibit positive differences of \(+0.142\) and \(+0.170\), with \(p=0.057\) and \(p=0.086\). These portfolios obtain lower CAGR than equal weight and substantially lower volatility. ICA+CTS ERC remains above equal weight by \(+0.100\), but ICA+CTS IRP has a difference of \(-0.077\). Neither ICA+CTS difference is statistically significant.

	In both SECTOR and MOM10, all model--construction combinations gain positive Sharpe differences relative to equal weight. The strongest SECTOR results are ICA+NTS ERC at \(+0.051\) with \(p=0.023\) and ICA+CTS ERC at \(+0.042\) with \(p=0.059\). In MOM10, ICA+NTS and ICA+CTS ERC obtain differences of \(+0.028\) with \(p=0.052\) and \(+0.032\) with \(p=0.035\). ICA+CTS IRP also has a difference of \(+0.026\) with \(p=0.076\). The ICA-based ERC portfolios combine slightly higher CAGR with lower volatility relative to equal weight.

	Table~\ref{tab:normal_benchmark_sig_all_lb12_h1_rp} compares the tempered stable portfolios with the matched Gaussian EVaR benchmark. In MOM10 ERC, ICA+NTS and ICA+CTS show gross Sharpe differences of \(+0.014\) with \(p=0.041\) and \(+0.019\) with \(p=0.057\), respectively. After 25-basis-point transaction costs, the differences fall to \(+0.008\) with \(p=0.218\) and \(+0.009\) with \(p=0.370\). Direct MNTS is below the Gaussian benchmark by 0.007 in both gross and net Sharpe ratio, with \(p=0.073\) and \(p=0.063\). Among the remaining comparisons, no gross or net Sharpe difference has a \(p\)-value below 0.10.

\begingroup
\begin{table}[!t]
\centering
\caption{Risk-parity portfolios relative to the matched Gaussian EVaR-contribution benchmark. The table reports the Gaussian gross and net Sharpe ratios and the Sharpe differences of the tempered-stable specifications, gross and net of a 25 basis point transaction-cost assumption. $^{*}$, $^{**}$, and $^{***}$ denote significance at the 10\%, 5\%, and 1\% levels. }
\label{tab:normal_benchmark_sig_all_lb12_h1_rp}
\scriptsize
\setlength{\tabcolsep}{4.2pt}
\resizebox{\linewidth}{!}{%
\begin{tabular}{llrr|rrrr|rrrr|rrrr}
\toprule
 &  & \multicolumn{2}{c|}{Normal benchmark} & \multicolumn{4}{c|}{MNTS} & \multicolumn{4}{c|}{ICA$+$NTS} & \multicolumn{4}{c}{ICA$+$CTS} \\
\cmidrule(lr){3-4}\cmidrule(lr){5-8}\cmidrule(lr){9-12}\cmidrule(lr){13-16}
Market & Construction & Gross SR & Net SR & $\Delta$ gross & $p$ & $\Delta$ net & $p$ & $\Delta$ gross & $p$ & $\Delta$ net & $p$ & $\Delta$ gross & $p$ & $\Delta$ net & $p$ \\
\midrule
SECTOR & IRP & 0.558 & 0.556 & -0.001 & 0.714 & -0.002 & 0.630 & +0.001 & 0.938 & -0.008 & 0.481 & -0.000 & 0.997 & -0.015 & 0.307 \\
 & ERC & 0.560 & 0.557 & -0.001 & 0.803 & -0.002 & 0.745 & +0.028 & 0.110 & +0.019 & 0.273 & +0.019 & 0.249 & +0.005 & 0.746 \\
\addlinespace
MOM10 & IRP & 0.584 & 0.582 & -0.002 & 0.784 & -0.002 & 0.660 & +0.004 & 0.514 & -0.002 & 0.728 & +0.009 & 0.279 & -0.002 & 0.838 \\
 & ERC & 0.581 & 0.579 & -0.007$^{*}$ & 0.073 & -0.007$^{*}$ & 0.063 & +0.014$^{**}$ & 0.041 & +0.008 & 0.218 & +0.019$^{*}$ & 0.057 & +0.009 & 0.370 \\
\addlinespace
XASSET & IRP & 0.760 & 0.753 & -0.005 & 0.750 & -0.006 & 0.658 & -0.039 & 0.402 & -0.059 & 0.212 & -0.259 & 0.322 & -0.286 & 0.282 \\
 & ERC & 0.827 & 0.814 & -0.003 & 0.834 & -0.004 & 0.798 & -0.079 & 0.331 & -0.101 & 0.224 & -0.150 & 0.384 & -0.183 & 0.291 \\
\bottomrule
\end{tabular}%
}
\end{table}
\endgroup

\subsection{Matched EVaR--CVaR comparisons}
\label{sec_full_sample_matched}
	We next compare matched EVaR and CVaR portfolios across the three return-model specifications and the two risk parity rules. The matched design holds the return-model specification and risk parity rule fixed and changes only the tail-risk measure. Table~\ref{tab:matched_pairs_sig_all_lb12_h1_rp} reports the matched full-sample comparisons. Under direct MNTS, the EVaR--CVaR Sharpe differences are small for both IRP and ERC in all three universes. The ERC differences are \(-0.001\) in SECTOR, \(-0.003\) in MOM10, and \(+0.006\) in XASSET, with \(p\)-values of 0.642, 0.179, and 0.651. None is statistically significant.

\begingroup
\begin{table}[!t]
\centering
\caption{Matched comparison of matched EVaR- and CVaR-based risk-parity portfolios with significance of the Sharpe difference. $^{*}$, $^{**}$, and $^{***}$ denote significance at the 10\%, 5\%, and 1\% levels. }
\label{tab:matched_pairs_sig_all_lb12_h1_rp}
\scriptsize
\setlength{\tabcolsep}{3.4pt}
\resizebox{\linewidth}{!}{%
\begin{tabular}{lllrrrrrrrrrr}
\toprule
 &  & \multicolumn{4}{c}{EVaR-based} & \multicolumn{4}{c}{CVaR-based} & & \\
\cmidrule(lr){4-7}\cmidrule(lr){8-11}
Market & Construction & Approach & CumRet & SR & MDD & TO & CumRet & SR & MDD & TO & $\Delta$SR & $p$ \\
\midrule
SECTOR & IRP & MNTS & 729.31 & 0.557 & 50.29 & 1.66 & 733.33 & 0.558 & 50.15 & 1.50 & -0.001 & 0.299 \\
 &  & ICA$+$NTS & 745.51 & 0.559 & 50.11 & 6.60 & 738.63 & 0.560 & 49.94 & 1.59 & -0.001 & 0.936 \\
 &  & ICA$+$CTS & 739.85 & 0.558 & 49.35 & 9.91 & 739.68 & 0.560 & 49.94 & 1.61 & -0.002 & 0.871 \\
\addlinespace
 & ERC & MNTS & 725.76 & 0.559 & 50.29 & 1.97 & 730.80 & 0.560 & 50.12 & 1.87 & -0.001 & 0.642 \\
 &  & ICA$+$NTS & 853.79 & 0.588 & 48.56 & 6.93 & 735.09 & 0.560 & 49.85 & 2.93 & +0.028 & 0.110 \\
 &  & ICA$+$CTS & 821.40 & 0.579 & 48.39 & 9.49 & 746.11 & 0.562 & 50.04 & 2.52 & +0.017 & 0.323 \\
\addlinespace
MOM10 & IRP & MNTS & 1698.10 & 0.583 & 58.80 & 1.94 & 1665.13 & 0.581 & 58.65 & 1.47 & +0.002 & 0.653 \\
 &  & ICA$+$NTS & 1754.58 & 0.588 & 58.42 & 5.39 & 1685.53 & 0.583 & 58.40 & 1.55 & +0.005 & 0.387 \\
 &  & ICA$+$CTS & 1815.86 & 0.593 & 58.85 & 8.52 & 1686.46 & 0.583 & 58.39 & 1.55 & +0.010 & 0.223 \\
\addlinespace
 & ERC & MNTS & 1619.55 & 0.574 & 59.02 & 1.73 & 1642.06 & 0.577 & 58.85 & 1.64 & -0.003 & 0.179 \\
 &  & ICA$+$NTS & 1840.93 & 0.595 & 58.53 & 5.27 & 1658.21 & 0.580 & 58.53 & 1.96 & +0.015$^{**}$ & 0.028 \\
 &  & ICA$+$CTS & 1897.58 & 0.600 & 58.76 & 8.09 & 1675.14 & 0.582 & 58.35 & 1.91 & +0.018$^{*}$ & 0.065 \\
\addlinespace
XASSET & IRP & MNTS & 185.06 & 0.755 & 24.44 & 1.98 & 187.89 & 0.760 & 24.13 & 1.68 & -0.005 & 0.421 \\
 &  & ICA$+$NTS & 189.92 & 0.720 & 25.46 & 6.78 & 187.87 & 0.766 & 23.33 & 1.69 & -0.045 & 0.321 \\
 &  & ICA$+$CTS & 146.52 & 0.501 & 47.11 & 11.19 & 188.39 & 0.767 & 23.15 & 1.78 & -0.266 & 0.318 \\
\addlinespace
 & ERC & MNTS & 184.16 & 0.824 & 20.57 & 3.04 & 187.15 & 0.818 & 20.52 & 3.28 & +0.006 & 0.651 \\
 &  & ICA$+$NTS & 187.95 & 0.748 & 21.97 & 8.52 & 187.64 & 0.809 & 21.30 & 3.98 & -0.061 & 0.440 \\
 &  & ICA$+$CTS & 184.02 & 0.678 & 34.56 & 12.76 & 190.46 & 0.804 & 20.93 & 3.02 & -0.126 & 0.474 \\
\bottomrule
\end{tabular}%
}
\end{table}
\endgroup

	The larger matched differences occur under ICA-based ERC. In MOM10, ICA+NTS EVaR-ERC has a Sharpe ratio of 0.595, compared with 0.580 for matched CVaR-ERC. The difference is \(+0.015\) with \(p=0.028\). ICA+CTS gives a difference of \(+0.018\) with \(p=0.065\). These are the only matched EVaR--CVaR Sharpe differences with \(p<0.10\) in Table~\ref{tab:matched_pairs_sig_all_lb12_h1_rp}.

	The ICA-based ERC differences change sign across universes. In SECTOR, the NTS and CTS differences are \(+0.028\) and \(+0.017\), compared with \(-0.061\) and \(-0.126\) in XASSET. None of these four differences has a \(p\)-value below 0.10. XASSET ICA+CTS also has a larger maximum drawdown under EVaR than under matched CVaR, 34.56\% versus 20.93\%.

	IRP exhibits a different pattern. None of the matched IRP Sharpe differences has a \(p\)-value below 0.10. The largest observed IRP difference is the XASSET ICA+CTS value of \(-0.266\), but its \(p\)-value is 0.318. The matched EVaR--CVaR differences with \(p<0.10\) appear only in the ICA-based ERC comparisons.

\subsection{Transaction costs and turnover}
\label{sec_cost_results}
	We next examine transaction costs and turnover for the matched EVaR--CVaR portfolios. Table~\ref{tab:cost_premium_sig_all_lb12_h1_rp} reports the matched EVaR--CVaR Sharpe differences at transaction costs of 0, 5, 10, and 25 basis points. The positive MOM10 ICA ERC differences become smaller as costs increase. ICA+NTS falls from \(+0.015\) with \(p=0.028\) to \(+0.010\) with \(p=0.140\), and ICA+CTS falls from \(+0.018\) with \(p=0.065\) to \(+0.009\) with \(p=0.373\). The positive point estimates remain at 25 basis points, but no EVaR--CVaR difference at that cost level has a \(p\)-value below 0.10.

\begingroup
\begin{table}[!t]
\centering
\caption{Transaction-cost path of the EVaR--CVaR Sharpe premium for the risk-parity constructions. Each entry is the Sharpe ratio of the EVaR-based portfolio minus that of the matched CVaR-based portfolio. The zero-cost column is the gross difference. $^{*}$, $^{**}$, and $^{***}$ denote significance at the 10\%, 5\%, and 1\% levels. }
\label{tab:cost_premium_sig_all_lb12_h1_rp}
\scriptsize
\setlength{\tabcolsep}{4pt}
\begin{tabular}{lllrrrrrrrr}
\toprule
Market & Construction & Approach & 0 bps & $p$ & 5 bps & $p$ & 10 bps & $p$ & 25 bps & $p$ \\
\midrule
SECTOR & IRP & MNTS & -0.001 & 0.299 & -0.002 & 0.282 & -0.002 & 0.264 & -0.002 & 0.222 \\
 &  & ICA$+$NTS & -0.001 & 0.936 & -0.003 & 0.812 & -0.004 & 0.689 & -0.010 & 0.381 \\
 &  & ICA$+$CTS & -0.002 & 0.871 & -0.005 & 0.715 & -0.008 & 0.570 & -0.017 & 0.231 \\
\addlinespace
 & ERC & MNTS & -0.001 & 0.642 & -0.001 & 0.633 & -0.001 & 0.624 & -0.002 & 0.593 \\
 &  & ICA$+$NTS & +0.028 & 0.110 & +0.027 & 0.127 & +0.026 & 0.148 & +0.021 & 0.227 \\
 &  & ICA$+$CTS & +0.017 & 0.323 & +0.015 & 0.394 & +0.012 & 0.476 & +0.005 & 0.771 \\
\addlinespace
MOM10 & IRP & MNTS & +0.002 & 0.653 & +0.002 & 0.671 & +0.002 & 0.692 & +0.001 & 0.753 \\
 &  & ICA$+$NTS & +0.005 & 0.387 & +0.004 & 0.516 & +0.002 & 0.661 & -0.001 & 0.857 \\
 &  & ICA$+$CTS & +0.010 & 0.223 & +0.008 & 0.334 & +0.006 & 0.477 & -0.001 & 0.944 \\
\addlinespace
 & ERC & MNTS & -0.003 & 0.179 & -0.003 & 0.174 & -0.003 & 0.168 & -0.003 & 0.157 \\
 &  & ICA$+$NTS & +0.015$^{**}$ & 0.028 & +0.014$^{**}$ & 0.040 & +0.013$^{*}$ & 0.055 & +0.010 & 0.140 \\
 &  & ICA$+$CTS & +0.018$^{*}$ & 0.065 & +0.016 & 0.101 & +0.014 & 0.147 & +0.009 & 0.373 \\
\addlinespace
XASSET & IRP & MNTS & -0.005 & 0.421 & -0.005 & 0.398 & -0.005 & 0.379 & -0.006 & 0.324 \\
 &  & ICA$+$NTS & -0.045 & 0.321 & -0.049 & 0.284 & -0.053 & 0.248 & -0.065 & 0.163 \\
 &  & ICA$+$CTS & -0.266 & 0.318 & -0.271 & 0.311 & -0.276 & 0.302 & -0.292 & 0.281 \\
\addlinespace
 & ERC & MNTS & +0.006 & 0.651 & +0.006 & 0.641 & +0.006 & 0.629 & +0.006 & 0.598 \\
 &  & ICA$+$NTS & -0.061 & 0.440 & -0.065 & 0.417 & -0.068 & 0.391 & -0.078 & 0.324 \\
 &  & ICA$+$CTS & -0.126 & 0.474 & -0.133 & 0.448 & -0.140 & 0.423 & -0.160 & 0.362 \\
\bottomrule
\end{tabular}
\end{table}
\endgroup

	Turnover differs more clearly under the ICA-based ERC portfolios. Under direct MNTS, EVaR and CVaR turnover are similar. For ICA+NTS ERC, EVaR and CVaR turnover are 6.93\% and 2.93\% in SECTOR, 5.27\% and 1.96\% in MOM10, and 8.52\% and 3.98\% in XASSET. ICA+CTS shows an even larger turnover gap. ICA-based EVaR-ERC has higher turnover than matched CVaR-ERC in all three universes.

\subsection{EVaR deviation and the fitted-location term}
\label{sec_centered_results}
	We finally assess whether removing the fitted location term materially changes EVaR-ERC portfolio performance. Table~\ref{tab:ercdev_matched_sig_all_lb12_h1_rp} compares raw EVaR-ERC and EVaR-deviation ERC portfolios. For the EVaR rows, none of the raw EVaR--EVaR-deviation Sharpe differences has a \(p\)-value below 0.10. The \(p\)-values range from 0.151 to 0.993. Direct MNTS changes little in all three universes. ICA+NTS also changes little in full-sample Sharpe terms.

\begingroup
\begin{table}[!t]
\centering
\caption{Matched comparison of the EVaR-deviation ERC and raw EVaR-ERC risk-parity portfolios with significance of the Sharpe difference. $^{*}$, $^{**}$, and $^{***}$ denote significance at the 10\%, 5\%, and 1\% levels. }
\label{tab:ercdev_matched_sig_all_lb12_h1_rp}
\scriptsize
\setlength{\tabcolsep}{3.4pt}
\resizebox{\linewidth}{!}{%
\begin{tabular}{lllrrrrrrrrrr}
\toprule
 &  & \multicolumn{4}{c}{EVaR-deviation ERC} & \multicolumn{4}{c}{raw EVaR-ERC} & & \\
\cmidrule(lr){4-7}\cmidrule(lr){8-11}
Market & Risk measure & Approach & CumRet & SR & MDD & TO & CumRet & SR & MDD & TO & $\Delta$SR & $p$ \\
\midrule
SECTOR & EVaR & MNTS & 724.94 & 0.558 & 50.32 & 2.01 & 725.76 & 0.559 & 50.29 & 1.97 & -0.001 & 0.512 \\
 &  & ICA$+$NTS & 870.71 & 0.592 & 48.58 & 6.78 & 853.79 & 0.588 & 48.56 & 6.93 & +0.004 & 0.265 \\
 &  & ICA$+$CTS & 819.68 & 0.579 & 48.41 & 9.42 & 821.40 & 0.579 & 48.39 & 9.49 & -0.000 & 0.993 \\
\addlinespace
MOM10 & EVaR & MNTS & 1685.96 & 0.580 & 59.09 & 2.21 & 1619.55 & 0.574 & 59.02 & 1.73 & +0.006 & 0.434 \\
 &  & ICA$+$NTS & 1837.88 & 0.595 & 58.59 & 5.35 & 1840.93 & 0.595 & 58.53 & 5.27 & -0.000 & 0.537 \\
 &  & ICA$+$CTS & 1869.45 & 0.597 & 58.86 & 8.05 & 1897.58 & 0.600 & 58.76 & 8.09 & -0.002 & 0.151 \\
\addlinespace
XASSET & EVaR & MNTS & 183.78 & 0.821 & 20.74 & 3.09 & 184.16 & 0.824 & 20.57 & 3.04 & -0.003 & 0.637 \\
 &  & ICA$+$NTS & 202.96 & 0.750 & 21.07 & 8.47 & 187.95 & 0.748 & 21.97 & 8.52 & +0.003 & 0.921 \\
 &  & ICA$+$CTS & 176.25 & 0.641 & 35.16 & 11.21 & 184.02 & 0.678 & 34.56 & 12.76 & -0.037 & 0.377 \\
\bottomrule
\end{tabular}%
}
\end{table}
\endgroup

	The largest observed centering change is ICA+CTS in XASSET, where the Sharpe ratio falls from 0.678 to 0.641. The reported difference is \(-0.037\) with \(p=0.377\). In MOM10, ICA+NTS is unchanged at the reported precision and ICA+CTS changes by only \(-0.002\). In SECTOR, the ICA+NTS Sharpe ratio increases by 0.004 with \(p=0.265\). These comparisons provide no evidence that removing the fitted location term materially changes full-sample EVaR-ERC Sharpe performance.

\section{Conclusion}
\label{sec_conclusion}
	We develop EVaR-based risk parity under MNTS and ICA-based tempered stable return models. We derive asset-level Euler contributions for EVaR and the corresponding EVaR-deviation risk contributions. Under ICA, EVaR uses exponentially tilted component means, while CVaR uses component means conditional on the portfolio lower-tail event. Centering separates the fitted location term from EVaR risk contributions, and Gaussian EVaR-deviation IRP and ERC give the same portfolio weights as conventional volatility IRP and ERC.

	Sharpe ratio differences relative to equal weight are positive for all ERC portfolios and for all but one IRP portfolio, with statistically significant improvements in several comparisons. Selected ICA-based ERC portfolios also show positive gross Sharpe differences relative to the Gaussian EVaR benchmark. Matched EVaR and CVaR portfolios remain close under direct MNTS. The matched differences are larger under ICA-based ERC, but their signs vary across universes. ICA-based EVaR-ERC has higher turnover than matched CVaR-ERC, and transaction costs weaken the positive matched Sharpe differences in MOM10. Centering has little effect on full-sample Sharpe ratios. These findings do not indicate a uniform performance advantage of EVaR over CVaR and should be interpreted within the investment universes, sample periods, and backtest design considered here.

\clearpage
\bibliographystyle{plainnat}
\bibliography{TSEVaRRP}

\end{document}